\documentclass[12pt,english]{article}
\usepackage[T1]{fontenc}
\usepackage[latin9]{inputenc}
\usepackage{amsmath}
\usepackage{amsthm}
\newtheorem{proposition}{Proposition}
\newtheorem{lemma}{Lemma}

\theoremstyle{definition}

\theoremstyle{remark}

\usepackage[authoryear,round]{natbib}
\usepackage{amssymb}
\usepackage[a4paper]{geometry}
\PassOptionsToPackage{normalem}{ulem}
\usepackage{ulem}
\usepackage{babel}
\usepackage{tikz}
\usepackage{pgfplots}
\usepackage{subcaption}
\pgfplotsset{compat=1.18}
\usepackage{xcolor}
\usepackage[
  colorlinks=true,
  linkcolor=black,
  citecolor=blue,
  urlcolor=blue
]{hyperref}

\begin{document}
\title{Two Equivalence Results between Unemployment Insurance and Wage Insurance}
\author{Anchi (Bryant) Xia\thanks{I thank Jan Morgan for helpful discussions that jumpstarted the note, Eric (``Big G'') Gao for encouraging me to actually finish writing this note, and Jessie Chen for putting up with me every time I bring this topic up in our conversations. I acknowledge generative AI use for proof writing and visual illustration. All mistakes are my own.}}

\maketitle

\section{Introduction}

In this short note, I show that in a \citet{mccall1970economics} model with risk-neutral agents, a system of wage insurance combined with unemployment insurance financed by a lump-sum tax on the employed can be replicated---in the \emph{ex ante} sense---by a system of unemployment insurance that depends on the agent's wage when last employed and a tax that depends on the agent's wage when employed. This holds with or without endogenous search when unemployed. I also show that wage insurance is not binding until the last earned wage exceeds a threshold, which may be of independent interest.

Loosely speaking, these equivalence results hold since wage insurance manipulates both the surplus from searching as well as the reservation wage, so the combination of unemployment insurance, which manipulates the reservation wage, and income taxes, which manipulates the surplus from searching, can be designed to replicate a transfer system that includes wage insurance.

A number of strong assumptions are made to establish my two equivalence results. The point is not to say that any wage insurance scheme can be replicated by a suitably chosen unemployment insurance and appropriate modifications to the income tax system, but rather to highlight the connection between wage insurance and unemployment insurance. I hope the results in this note can alert the reader that tools (and results) developed for the analysis of unemployment insurance may well apply to the study of wage insurance. More importantly, I hope these results provide benchmarks so that researchers can identify precisely through which economic channel wage insurance offers an improvement in their setting over and beyond what can be achieved by a cocktail of existing instruments.

The rest of this note is organized as follows. Section 2 introduces the model. Section 3 establishes the equivalence with exogenous job arrival when unemployed. Section 4 establishes the equivalence when search effort is endogenous.

\section{Model}

The model is a version of \citet{mccall1970economics}. Time is continuous. At $t=0$, all workers are unemployed with previous wage distributed according to $H$. The distribution of wage offers is denoted by $F$. Both distributions are assumed to be continuously differentiable. Employment is absorbing after acceptance. We now introduce the wage insurance before turning to the worker's problem.

\paragraph*{Wage insurance}

Let $z$ and $w$ denote the worker's past and current wages, respectively. The worker's net-of-tax consumption is then given by\footnote{Same formulation as \citet{hyman2024wage}.}
\begin{equation}
g_{\phi}\left(w,z\right)=w+\phi\left(z-w\right)_{+}-T=\begin{cases}
\left(1-\phi\right)w+\phi z-T & w\leq z\\
w-T & w>z
\end{cases}
\end{equation}
Here $x_{+}\equiv\max\left\{ x,0\right\} $ and $\phi \geq 0$ governs the generosity of the wage insurance system.

\paragraph*{Worker's problem}

The unemployed worker's indirect utility is defined via the following HJB equation:
\begin{equation}
rV^{U}\left(z\right)=\max_{\lambda}b\left(z\right)-\gamma\psi\left(\lambda\right)+\left(\lambda+\uline{\lambda}\right)\int_{\uline{w}_{\phi}\left(z\right)}\left[\frac{g_{\phi}\left(w,z\right)}{r}-V^{U}\left(z\right)\right]dF\left(w\right)
\end{equation}
Here $b\left(z\right)$ is the unemployment benefit, $\psi\left(\lambda\right)$ is the cost of searching, $\uline{\lambda}$ is a job arrival rate even if no search effort is exerted, and $\uline{w}_{\phi}\left(z\right)$ is the reservation wage for an agent with last wage $z$. $\psi$ is assumed to be strictly increasing and strictly convex. $\uline{w}_{\phi}\left(z\right)$ satisfies
\begin{equation}
\frac{g_{\phi}\left(\uline{w}_{\phi}\left(z\right),z\right)}{r}=V^{U}\left(z\right).
\end{equation}
In the no search effort case, I allow $b\left(z\right)$ to be flexible; in the case with endogenous search effort, I hold $b\left(z\right)$ to be constant. Lastly, $\gamma$ is a parameter governing the search process. $\gamma=\infty$ corresponds to the case of exogenous Poisson job arrival with rate $\uline{\lambda}$. $\gamma=1$ and $\uline{\lambda}=0$ nests conventional models with search effort.

\section{No search effort}

Let us first examine the case when the job arrival is exogenous, corresponding to the case where $\gamma=\infty$ so the optimal $\lambda$ is $0$. 

\begin{proposition}[Exogenous search effort]
    Suppose $\gamma=\infty$, so job offers arrive at the exogenous Poisson rate $\uline{\lambda}$. Fix any budget-balanced UI/WI policy $\left(b\left(\cdot\right),T,\phi\right)$ that induces reservation wage $\uline{w}_{\phi}\left(z\right)$ for each previous wage $z$. There exists a UI-only policy with previous-wage-contingent benefit $b^{*}\left(z\right)$ and a lump-sum current wage tax $T^{*}$ that induces the same reservation rule, satisfies the government budget in expectation at $t=0$, and delivers the same ex ante welfare.
\end{proposition}

\begin{proof}

The proof proceeds as follows. First, from necessary conditions, we will characterize the unemployment benefit schedule $b^{*}\left(z\right)$ needed to replicate the induced reservation $\uline{w}_{\phi}\left(z\right)$. We then show that under $b^{*}\left(z\right)$ and taxes $T^{*}$, $\uline{w}_{\phi}\left(z\right)$ is indeed the optimal solution. We conclude the argument by establishing ex ante welfare equivalence via an explicit calculation.

\paragraph*{Backing out $b^{*}\left(z\right)$}

We wish to first match the acceptance wage. Given unemployment benefit $b^{*}\left(z\right)$ and taxes $T^{*}$, if the acceptance wage were to be $\uline{w}_{\phi}\left(z\right)$, then the following relation must hold:
\begin{equation}
\uline{w}_{\phi}\left(z\right)-T^{*}=b^{*}\left(z\right)+\uline{\lambda}\int_{\uline{w}_{\phi}\left(z\right)}\left[\frac{w-\uline{w}_{\phi}\left(z\right)}{r}\right]dF\left(w\right).
\end{equation}
Rearranging, we have
\begin{equation}
b^{*}\left(z\right)=\uline{w}_{\phi}\left(z\right)-T^{*}-\uline{\lambda}\int_{\uline{w}_{\phi}\left(z\right)}\left[\frac{w-\uline{w}_{\phi}\left(z\right)}{r}\right]dF\left(w\right),
\end{equation}
The lump-sum tax on the employed, $T^{*}$, is in turn pinned down by budget balancing; we formally clarify its existence in the last step when showing ex ante welfare equivalence.

\paragraph*{Verifying the reservation wage match}

Given $b^{*}\left(z\right)$ and $T^{*}$, we now wish to verify that $\uline{w}_{\phi}\left(z\right)$ is indeed the reservation wage an agent with last wage $z$ would choose. To see this, suppose the reservation wage were to be $y$, then it solves
\begin{equation}
y-T^{*}-b^{*}\left(z\right)-\uline{\lambda}\int_{y}\left[\frac{w-y}{r}\right]dF\left(w\right)=0.
\end{equation}
It is straightforward to show that the derivative of the LHS w.r.t. $y$ is strictly positive, so since $\uline{w}_{\phi}\left(z\right)$ solves this equation, it is the unique solution. Since the verification is for an arbitrary $z$, we have shown the replication for all $z$.

\paragraph*{Ex ante welfare equivalence}

Having shown that we can replicate the reservation wage, it remains to establish ex ante welfare equivalence as well as the existence of a $T^{*}$ that balances the budget. Let us first calculate the ex ante welfare of the economy with no wage insurance. The date $0$ welfare is given by
\begin{equation}
\int\int_{0}^{\infty}\alpha\left(z\right)e^{-\alpha\left(z\right)t}\left\{ \int_{0}^{t}e^{-rs}b^{*}\left(z\right)ds+e^{-rt}\int_{\uline{w}_{\phi}\left(z\right)}\frac{w-T^{*}}{r}\frac{dF\left(w\right)}{1-F\left(\uline{w}_{\phi}\left(z\right)\right)}\right\} dtdH\left(z\right),
\end{equation}
where $\alpha\left(z\right)\equiv\uline{\lambda}\left(1-F\left(\uline{w}_{\phi}\left(z\right)\right)\right)$ is the Poisson arrival rate of an accepted job. Next, budget balancing requires
\begin{equation}
\int\int_{0}^{\infty}\alpha\left(z\right)e^{-\alpha\left(z\right)t}\left\{ -\int_{0}^{t}e^{-rs}b^{*}\left(z\right)ds+e^{-rt}\frac{T^{*}}{r}\right\} dtdH\left(z\right)=0.
\end{equation}
Since adjusting $T^{*}$ has no bearing on $\alpha\left(z\right)$ given our construction, it is easy to see there exists a unique $T^{*}$ satisfying the budget balancing condition above. Using the budget balancing relation, ex ante welfare is equivalently
\begin{equation}
\int\int_{0}^{\infty}\alpha\left(z\right)e^{-\alpha\left(z\right)t}e^{-rt}\int_{\uline{w}_{\phi}\left(z\right)}\frac{w}{r}\frac{dF\left(w\right)}{1-F\left(\uline{w}_{\phi}\left(z\right)\right)}dtdH\left(z\right).
\end{equation}
More concisely, ex ante welfare can be written as
\begin{equation}
\int\frac{\alpha\left(z\right)}{\alpha\left(z\right)+r}\frac{\mathbb{E}\left[W\bigg|W\geq\uline{w}_{\phi}\left(z\right)\right]}{r}dH\left(z\right).
\end{equation}
Next, we turn to the economy with wage insurance, the ex ante welfare is given by
\begin{equation}
\int\int_{0}^{\infty}\alpha\left(z\right)e^{-\alpha\left(z\right)t}\left\{ \int_{0}^{t}e^{-rs}b\left(z\right)ds+e^{-rt}\int_{\uline{w}_{\phi}\left(z\right)}\frac{g_{\phi}\left(w,z\right)}{r}\frac{dF\left(w\right)}{1-F\left(\uline{w}_{\phi}\left(z\right)\right)}\right\} dtdH\left(z\right).
\end{equation}
Budget balancing requires
{\footnotesize\begin{equation}
\int\int_{0}^{\infty}\alpha\left(z\right)e^{-\alpha\left(z\right)t}\left\{ -\int_{0}^{t}e^{-rs}b\left(z\right)ds+e^{-rt}\int_{\uline{w}_{\phi}\left(z\right)}\frac{T-\phi\left(z-w\right)_{+}}{r}\frac{dF\left(w\right)}{1-F\left(\uline{w}_{\phi}\left(z\right)\right)}\right\} dtdH\left(z\right)=0.
\end{equation}}
So once again, the ex ante welfare simplifies to
\begin{equation}
\int\frac{\alpha\left(z\right)}{\alpha\left(z\right)+r}\frac{\mathbb{E}\left[W\bigg|W\geq\uline{w}_{\phi}\left(z\right)\right]}{r}dH\left(z\right).
\end{equation}
This concludes the proof.
\end{proof}

\section{With search effort}

We now turn to the case with search effort. First, we will establish a few useful properties in the economy featuring wage insurance. Let 
\begin{equation}
S_{\phi}\left(z\right)\equiv\int_{\uline{w}_{\phi}\left(z\right)}\left[\frac{g_{\phi}\left(w,z\right)}{r}-V^{U}\left(z\right)\right]dF\left(w\right)
\end{equation}
denote the search surplus for an agent whose last wage is $z$,

\paragraph*{Pooling and Monotonicity}

\begin{lemma} 
The following statements regarding the equilibrium with wage insurance are true:
\begin{enumerate}
\item If $z\leq x_{0}$ for $x_{0}$ to be defined, then $\uline{w}_{\phi}\left(z\right)=x_{0}$. Otherwise, for $z>x_{0}$, $\uline{w}_{\phi}\left(z\right)<z$ and $\uline{w}_{\phi}^{\prime}\left(z\right)<0$.
\item $S_{\phi}^{\prime}\left(z\right)=0$ if $z\leq x_{0}$ and $S_{\phi}^{\prime}\left(z\right)>0$ if $z>x_{0}$.
\end{enumerate} 
\end{lemma}

Effectively, this lemma says that wage insurance only binds for those with a last earned wage that is high enough, and that search surplus is flat prior to this cutoff but strictly increases as a function of the last earned wage after this cutoff. These features are illustrated in Figure 1 below. We now turn to proving these claims.

\begin{figure}[h]
\centering

\begin{subfigure}[t]{0.48\textwidth}
\centering
\begin{tikzpicture}
\begin{axis}[
    width=\textwidth,
    height=0.72\textwidth,
    axis lines=left,
    xmin=0, xmax=8,
    ymin=0, ymax=5.2,
    xlabel={Last wage $z$},
    ylabel={Reservation wage $\underline w_\phi(z)$},
    xtick={3},
    xticklabels={$x_0$},
    ytick={3},
    yticklabels={$x_0$},
    tick style={black},
    axis line style={black},
    label style={font=\small},
    tick label style={font=\small},
    clip=false,
]
\fill[blue!5] (axis cs:3,0) rectangle (axis cs:8,5.2);

\addplot[
    domain=0:5.2,
    samples=2,
    dashed,
    gray!65,
    line width=0.8pt
] {x};

\addplot[
    domain=0:3,
    samples=2,
    black,
    line width=1.3pt
] {3};

\addplot[
    domain=3:8,
    samples=100,
    black,
    line width=1.3pt
] {3 - 1.05*(1 - exp(-0.55*(x-3)))};

\addplot[
    dashed,
    gray!60,
    line width=0.8pt
] coordinates {(3,0) (3,3)};

\node[gray!70, font=\scriptsize, anchor=west] at (axis cs:4.75,4.9) {$\underline w=z$};
\node[font=\scriptsize, anchor=south] at (axis cs:1.45,3.08) {pooling};
\node[font=\scriptsize, anchor=south west, align=left] at (axis cs:4.05,2.45)
    {$\underline w_\phi'(z)<0$};

\end{axis}
\end{tikzpicture}
\caption{Reservation wage}
\end{subfigure}
\hfill
\begin{subfigure}[t]{0.48\textwidth}
\centering
\begin{tikzpicture}
\begin{axis}[
    width=\textwidth,
    height=0.72\textwidth,
    axis lines=left,
    xmin=0, xmax=8,
    ymin=0, ymax=5.2,
    xlabel={Last wage $z$},
    ylabel={Search surplus $S_\phi(z)$},
    xtick={3},
    xticklabels={$x_0$},
    ytick={2},
    yticklabels={$S_\phi(x_0)$},
    tick style={black},
    axis line style={black},
    label style={font=\small},
    tick label style={font=\small},
    clip=false,
]
\fill[blue!5] (axis cs:3,0) rectangle (axis cs:8,5.2);

\addplot[
    domain=0:3,
    samples=2,
    black,
    line width=1.3pt
] {2};

\addplot[
    domain=3:8,
    samples=100,
    black,
    line width=1.3pt
] {2 + 2.15*(1 - exp(-0.45*(x-3)))};

\addplot[
    dashed,
    gray!60,
    line width=0.8pt
] coordinates {(3,0) (3,2)};

\node[font=\scriptsize, anchor=south] at (axis cs:1.45,2.08) {$S_\phi'(z)=0$};
\node[font=\scriptsize, anchor=south west, align=left] at (axis cs:4.15,3.2)
    {$S_\phi'(z)>0$};

\end{axis}
\end{tikzpicture}
\caption{Search surplus}
\end{subfigure}

\caption{Illustration of the pooling and monotonicity lemma. For $z\leq x_0$, wage insurance does not bind, the reservation wage is constant, and search surplus is flat. For $z>x_0$, wage insurance binds, the reservation wage declines, and search surplus increases.}
\end{figure}
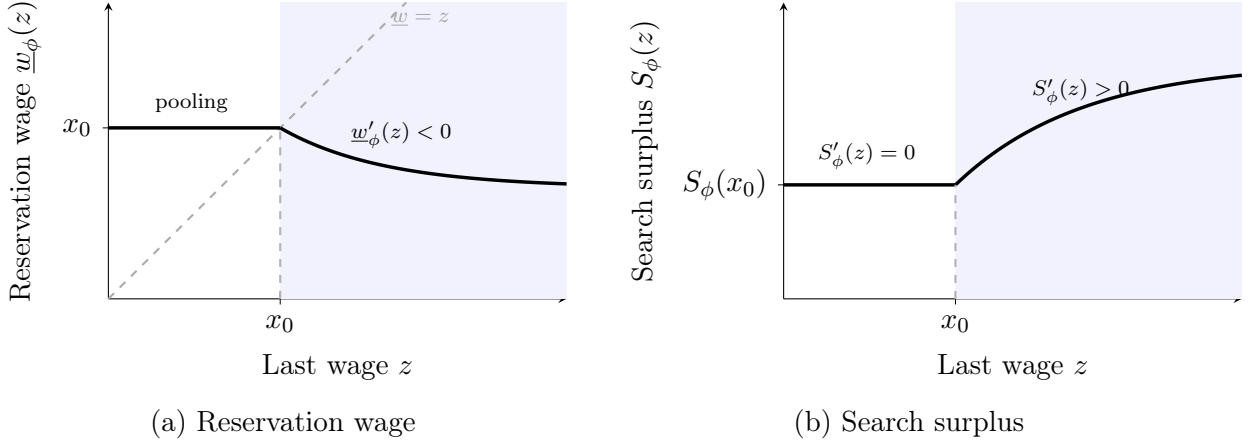

\begin{proof}
We proceed by establishing each claim separately.
\paragraph*{Claim 1}
Let us start by establishing the first claim. Define
\begin{align}
S_{\phi}\left(x,z\right) & \equiv\int_{x}\left[\frac{g_{\phi}\left(w,z\right)-g_{\phi}\left(x,z\right)}{r}\right]dF\left(w\right)\\
R\left(S\right) & =\max_{\lambda}-\psi\left(\lambda\right)+\lambda S\\
\Omega_{\phi}\left(x,z\right) & =g_{\phi}\left(x,z\right)-b-R\left(S_{\phi}\left(x,z\right)\right)
\end{align}
Here $S_{\phi}\left(x,z\right)$ is the search surplus given acceptance wage $x$ and past wage $z$; $R\left(S\right)$ is the expected value form searching net of search costs; and $\Omega_{\phi}\left(x,z\right)$ is the residual of not searching versus searching if the agent accepts $x$ as the reservation wage. Furthermore, define $x_{0}$ as the solution to the following equation
\begin{equation}
x_{0}-T-b-R\left(\int_{x_{0}}\frac{w-x_{0}}{r}dF\left(w\right)\right)=0.
\end{equation}
Here $x_{0}$ can be interpreted as the acceptance wage absent wage insurance altogether. It is easy to verify $x_{0}$ is uniquely defined under our assumption.\footnote{$x_{0}$ as defined may be outside of the support of $F$, but this is fine for our purpose.} Now, if $z\leq x_{0}$, then every accepted wage satisfies $w\geq x_{0}\geq z$, so wage insurance in fact does not bind, and by uniqueness we must have
\begin{equation}
\uline{w}_{\phi}\left(z\right)=x_{0}\quad\forall z\leq x_{0}.
\end{equation}
If $z>x_{0}$, we claim $\uline{w}_{\phi}\left(z\right)<z$. Otherwise, if $\uline{w}_{\phi}\left(z\right)\geq z$, wage insurance is again not binding, so by uniqueness we have $\uline{w}_{\phi}\left(z\right)=x_{0}$, a contradiction. Next we wish establish that $\uline{w}_{\phi}^{\prime}\left(z\right)<0$ provided $z>x_{0}$. Observe that for this range of past wages, we have
\begin{align}
S_{\phi x}\left(\uline{w}_{\phi}\left(z\right),z\right) & =-\frac{1}{r}\left(1-\phi\right)\left(1-F\left(\uline{w}_{\phi}\left(z\right)\right)\right)\\
S_{\phi z}\left(\uline{w}_{\phi}\left(z\right),z\right) & =-\frac{1}{r}\phi\left(1-F\left(z\right)\right)
\end{align}
Implicitly differentiating the following relation 
\begin{equation}
\Omega_{\phi}\left(\uline{w}_{\phi}\left(z\right),z\right)=0,
\end{equation}
we obtain
\begin{equation}
\uline{w}_{\phi}^{\prime}\left(z\right)=-\frac{R^{\prime}\left(S_{\phi}\left(\uline{w}_{\phi}\left(z\right),z\right)\right)\phi\left(1-F\left(z\right)\right)+r\phi}{r\left(1-\phi\right)+R^{\prime}\left(S_{\phi}\left(\uline{w}_{\phi}\left(z\right),z\right)\right)\left(1-\phi\right)\left(1-F\left(\uline{w}_{\phi}\left(z\right)\right)\right)}<0
\end{equation}
since $R^{\prime}=\lambda\geq$0 by the envelope theorem. This concludes our proof for the first claim.

\paragraph*{Claim 2}

The first part of the second claim is immediate given the first claim. It remains to show $S_{\phi}^{\prime}\left(z\right)>0$ if $z>x_{0}$. To see this,
{\footnotesize\begin{align}
S_{\phi}^{\prime}\left(z\right) & =S_{\phi x}\left(\uline{w}_{\phi}\left(z\right),z\right)\uline{w}_{\phi}^{\prime}\left(z\right)+S_{\phi z}\left(\uline{w}_{\phi}\left(z\right),z\right)\\
 & =\frac{1}{r}\left(1-\phi\right)\left(1-F\left(\uline{w}_{\phi}\left(z\right)\right)\right)\frac{R^{\prime}\left(S_{\phi}\left(\uline{w}_{\phi}\left(z\right),z\right)\right)\phi\left(1-F\left(z\right)\right)+r\phi}{r\left(1-\phi\right)+R^{\prime}\left(S_{\phi}\left(\uline{w}_{\phi}\left(z\right),z\right)\right)\left(1-\phi\right)\left(1-F\left(\uline{w}_{\phi}\left(z\right)\right)\right)}-\frac{1}{r}\phi\left(1-F\left(z\right)\right)
\end{align}}
Let $A\equiv1-F\left(\uline{w}_{\phi}\left(z\right)\right)$ and $B\equiv1-F\left(z\right)$. On the active region, $\uline{w}_{\phi}\left(z\right)<z$, so $A>B$. The expression can be rewritten as
\begin{equation}
S_{\phi}^{\prime}\left(z\right)=\frac{1}{r}\frac{\left(1-\phi\right)r\phi\left(A-B\right)}{r\left(1-\phi\right)+R^{\prime}\left(S_{\phi}\left(\uline{w}_{\phi}\left(z\right),z\right)\right)\left(1-\phi\right)A}>0,
\end{equation}
where the inequality follows from $R^{\prime}\left(S_{\phi}\left(\uline{w}_{\phi}\left(z\right),z\right)\right)>0$ and $A>B$.
\end{proof}

Armed with this result, we are ready to establish the equivalence in the case with endogenous search. 

\begin{proposition}[Endogenous search effort]
    Suppose the primitive assumptions in Section 1 hold and $\gamma=1$ and $\uline{\lambda}=0$, so the search effort is endogenous. Fix any budget-balanced UI/WI policy $\left(b,T,\phi\right)$ that induces reservation wage $\uline{w}_{\phi}\left(z\right)$ and search intensity $\lambda_{\phi}\left(z\right)$ for each previous wage $z$. There exists a UI-only policy with previous-wage-contingent benefit $b^{*}\left(z\right)$ and current-wage tax $\tau^{*}\left(w\right)$ that induces the same reservation rule and search intensity, satisfies the government budget in expectation at $t=0$, and delivers the same ex ante welfare.
\end{proposition}

\begin{proof}
The proof largely mirrors the earlier proposition. First, we will construct a consumption function ($q\left(w\right)=w-\tau^{*}\left(w\right))$ to match the search surplus at every reservation wage. Next, from necessary conditions, we deduce an accompanying schedule $b^{*}\left(z\right)$. We then verify that under our proposed system of unemployment insurance and taxes, the resulting reservation wage and search surplus for each type $z$ indeed matches those in the baseline economy featuring wage insurance. We conclude the proof by establishing welfare equivalence. 

\paragraph*{Constructing the consumption function}

The agent whose last wage is $z$ has the following HJB equation:
\begin{equation}
rV^{U}\left(z\right)=\max_{\lambda}b\left(z\right)-\psi\left(\lambda\right)+\lambda\int_{\uline{w}_{\phi}\left(z\right)}\left[\frac{g_{\phi}\left(w,z\right)}{r}-V^{U}\left(z\right)\right]dF\left(w\right),
\end{equation}
where 
\begin{equation}
\frac{g_{\phi}\left(\uline{w}_{\phi}\left(z\right),z\right)}{r}=V^{U}\left(z\right).
\end{equation}
Let $S_{\phi}\left(z\right)\equiv\int_{\uline{w}_{\phi}\left(z\right)}\left[\frac{g_{\phi}\left(w,z\right)}{r}-V^{U}\left(z\right)\right]dF\left(w\right)$ denote the surplus from searching in this economy. The agent's FOC pinning down search effort is given by
\begin{equation}
\psi^{\prime}\left(\lambda\right)=S_{\phi}\left(z\right).
\end{equation}
Turning to an economy with only a flexible tax system $\tau^{*}\left(w\right)$ and an unemployment insurance $b^{*}\left(z\right)$ dependent on the last wage, the HJB is
\begin{equation}
r\tilde{V}^{U}\left(z\right)=\max_{\lambda}b^{*}\left(z\right)-\psi\left(\lambda\right)+\lambda\int_{\uline{\tilde{w}}\left(z\right)}\left[\frac{q\left(w\right)}{r}-\tilde{V}^{U}\left(z\right)\right]dF\left(w\right),
\end{equation}
where $q\left(w\right)=w-\tau^{*}\left(w\right)$ and 
\begin{equation}
\frac{q\left(\uline{\tilde{w}}\left(z\right)\right)}{r}=\tilde{V}^{U}\left(z\right).
\end{equation}
Let $\tilde{S}\left(z\right)\equiv\int_{\uline{\tilde{w}}\left(z\right)}\left[\frac{q\left(w\right)}{r}-\tilde{V}^{U}\left(z\right)\right]dF\left(w\right)$ denote the search surplus in this economy. The agent's FOC is once again
\begin{equation}
\psi^{\prime}\left(\lambda\right)=\tilde{S}\left(z\right).
\end{equation}
If we can match the reservation wages $\uline{\tilde{w}}\left(z\right)=\uline{w}_{\phi}\left(z\right)$ and search surpluses $\tilde{S}\left(z\right)=S_{\phi}\left(z\right)$, then we have shown that individuals of all types exhibit the same search and acceptance behavior in both economies. We begin by constructing a consumption function such that
\begin{equation}
\int_{\uline{w}_{\phi}\left(z\right)}\frac{q\left(w\right)-q\left(\uline{w}_{\phi}\left(z\right)\right)}{r}dF\left(w\right)=S_{\phi}\left(z\right).
\end{equation}
That is, at every reservation wage in the baseline economy, $\uline{w}_{\phi}\left(z\right)$, if the reservation wage were to be equivalent in the counterfactual economy, then the search surpluses also agree. Additionally, as will be made clear in the verification step, we would also like $q$ to be strictly increasing. 

Guided by our earlier lemma, we shall split the construction separately for $z\leq x_{0}$ and $z>x_{0}$. First, when $z>x_{0}$, differentiating the desired relation, we have
\begin{equation}
-\left(1-F\left(\uline{w}_{\phi}\left(z\right)\right)\right)q^{\prime}\left(\uline{w}_{\phi}\left(z\right)\right)\uline{w}_{\phi}^{\prime}\left(z\right)=rS_{\phi}^{\prime}\left(z\right)
\end{equation}
So
\begin{equation}
q^{\prime}\left(\uline{w}_{\phi}\left(z\right)\right)=-\frac{rS_{\phi}^{\prime}\left(z\right)}{\uline{w}_{\phi}^{\prime}\left(z\right)\left(1-F\left(\uline{w}_{\phi}\left(z\right)\right)\right)}>0,
\end{equation}
where the inequality follows from our earlier lemma. Next, on the range where $z\leq x_{0}$, again by the earlier lemma, we know $\uline{w}_{\phi}\left(z\right)=x_{0}$ and the search surplus is a constant, so
\begin{equation}
\int_{x_{0}}\left(q\left(w\right)-q\left(x_{0}\right)\right)dF\left(w\right)=\int_{x_{0}}\left(w-x_{0}\right)dF\left(w\right)
\end{equation}
We can simply choose an extension to the $q$ function of the form
\begin{equation}
q\left(w\right)=q\left(x_{0}\right)+f\left(w-x_{0}\right)\quad\forall w\geq x_{0}
\end{equation}
where $f$ is any strictly increasing and differentiable function such that $f(0)=0, f^\prime(0)=0$, and
\begin{equation}
\int_{x_{0}}f\left(w-x_{0}\right)dF\left(w\right)=\int_{x_{0}}\left(w-x_{0}\right)dF\left(w\right).
\end{equation}
All told, we have implicitly defined $q$ and therefore $\tau^{*}$ over $\left[\uline{w}_{\phi}\left(\bar{z}\right),\bar{w}\right]$ where $\bar{z}$ is the supremum of the support of $H$ and $\bar{w}$ is the supremum of the support of $W$. If $\uline{w}_{\phi}\left(\bar{z}\right)>\uline{w}$, where $\uline{w}$ is the infimum of the support of $W$, we can choose any extension so that $q$ remains a strictly increasing and differentiable function. This completes the construction of our consumption function as well as the implied tax schedule.

\paragraph*{Recovering the unemployment benefits}

To recover the schedule of unemployment benefits, let $\uline{\tilde{w}}\left(z\right)=\uline{w}_{\phi}\left(z\right)$. Then,
\begin{equation}
q\left(\uline{w}_{\phi}\left(z\right)\right)=\max_{\lambda}b^{*}\left(z\right)-\psi\left(\lambda\right)+\lambda\int_{\uline{w}_{\phi}\left(z\right)}\left[\frac{q\left(w\right)-q\left(\uline{w}_{\phi}\left(z\right)\right)}{r}\right]dF\left(w\right),
\end{equation}
Rearranging and letting $\lambda\left(z;q\right)$ denote the maximized choice of search effort under our constructed consumption function $q$, we have
\begin{equation}
b^{*}\left(z\right)=q\left(\uline{w}_{\phi}\left(z\right)\right)+\psi\left(\lambda\left(z;q\right)\right)-\lambda\left(z;q\right)\int_{\uline{w}_{\phi}\left(z\right)}\left[\frac{q\left(w\right)-q\left(\uline{w}_{\phi}\left(z\right)\right)}{r}\right]dF\left(w\right).
\end{equation}

\paragraph*{Verifying the reservation wage and the search effort match for each type}

Up to level, we have pinned down unemployment insurance and taxes. The level indeterminacy comes from the fact that we can add the same constant to both schedules, which leaves the HJB unchanged and the search surplus unchanged. We will return to the level determination in the budget balancing step. For now, we wish to verify that under our constructed consumption function $q$ and benefits $b^{*}$, we indeed induce the same reservation wage and search effort. First, let us verify the reservation wage indeed matches. Suppose the reservation wage were to be $y$, then the optimal search effort is
\begin{equation}
\tilde{\lambda}\left(y\right)=\arg\max_{\lambda}-\psi\left(\lambda\right)+\lambda\int_{y}\left[\frac{q\left(w\right)-q\left(y\right)}{r}\right]dF\left(w\right)
\end{equation}
and the optimal reservation wage solves
\begin{equation}
q\left(y\right)-b^{*}\left(z\right)+\psi\left(\tilde{\lambda}\left(y\right)\right)-\tilde{\lambda}\left(y\right)\int_{y}\left[\frac{q\left(w\right)-q\left(y\right)}{r}\right]dF\left(w\right)=0.
\end{equation}
If $q$ is strictly increasing, then the derivative of RHS w.r.t. $y$ is
\begin{equation}
q^{\prime}\left(y\right)\left(1+\tilde{\lambda}\left(y\right)\frac{1-F\left(y\right)}{r}\right)+\tilde{\lambda}^{\prime}\left(y\right)\underbrace{\left(\psi^{\prime}\left(\tilde{\lambda}\left(y\right)\right)-\int_{y}\left[\frac{q\left(w\right)-q\left(y\right)}{r}\right]\right)}_{=0}>0,
\end{equation}
where the inequality follows from $q^{\prime}>0$. Since $y=\uline{w}_{\phi}\left(z\right)$ solves the equation, it is the unique solution. Having verified for an arbitrary $z$, we conclude that $\uline{\tilde{w}}\left(z\right)=\uline{w}_{\phi}\left(z\right)$. The equivalence on search effort is implied by our construction.

\paragraph*{Ex ante welfare equivalence}

Establishing ex ante welfare equivalence follows the same steps as proposition 1, except (i) the tax schedule in our modified economy is dependent on income, (ii) $\alpha\left(z\right)$ needs to be redefined as $\alpha\left(z\right)=\lambda\left(z\right)\left(1-F\left(\uline{w}_{\phi}\left(z\right)\right)\right)$, where $\lambda\left(z\right)=\lambda_\phi(z)$ is the optimized search effort we have matched across the two economies already, and (iii) we need to subtract off the search cost. Finally, regarding the normalization constant for our constructed tax and unemployment schedules, note that budget balancing requires
{\footnotesize \begin{equation}
\int\int_{0}^{\infty}\alpha\left(z\right)e^{-\alpha\left(z\right)t}\left\{ -\int_{0}^{t}e^{-rs}\left(b^{*}\left(z\right)+C\right)ds+e^{-rt}\int_{\uline{w}_{\phi}\left(z\right)}\frac{\tau^{*}\left(w\right)-C}{r}\frac{dF\left(w\right)}{1-F\left(\uline{w}_{\phi}\left(z\right)\right)}\right\} dtdH\left(z\right)=0.
\end{equation}}
Since we have argued that $\alpha\left(z\right)$ is invariant to $C$, the unique level of $C$ is pinned down by the budget balancing equation above. This concludes the proof.
\end{proof}

\bibliographystyle{plainnat}
\bibliography{references}

\end{document}